\documentclass[11pt]{article}
\usepackage[english]{babel}
\usepackage{amsmath,amsthm}
\usepackage{amsfonts}
\usepackage[comma,numbers,square,sort&compress]{natbib}
\usepackage{graphicx,subfigure,amsmath,amssymb,mathrsfs,amsfonts,amstext,amsthm}
\usepackage{latexsym, amssymb}
\usepackage{graphicx}
\usepackage{subfigure}
\usepackage{pgf,tikz}
\usepackage{pstricks}
\usepackage{hyperref}
\usepackage{graphicx}
\usepackage{subfigure}
\graphicspath{{./Figures/}}
\usepackage[utf8]{inputenc}
\usepackage[T1]{fontenc}

\newtheorem{thm}{Theorem}[section]

\newtheorem{lem}[thm]{Lemma}
\newtheorem{prop}[thm]{Proposition}
\theoremstyle{definition}
\newtheorem{defn}[thm]{Definition}
\theoremstyle{remark}
\newtheorem{rem}[thm]{Remark}
\numberwithin{equation}{section}
\newtheorem{example}{Example}
\begin{document}

\title{\bfseries\textrm{Memory-Driven Bounded Confidence Opinion Dynamics: A Hegselmann-Krause Model Based on Fractional-Order Methods}
\footnotetext{*
All the authors are with School of Mathematics and Statistics, Beijing Jiaotong University, Beijing 100044, China, {\tt 23121692@bjtu.edu.cn, su.wei@bjtu.edu.cn,
gjren@bjtu.edu.cn, ygyu@bjtu.edu.cn.} }             
}

\author{Meiru Jiang, Wei Su, Guojian Ren, Yongguang Yu}

%
\date{}%
\maketitle
\begin{abstract}
Memory effects play a crucial role in social interactions and decision-making processes. This paper proposes a novel fractional-order bounded confidence opinion dynamics model to characterize the memory effects in system states. Building upon the Hegselmann-Krause framework and fractional-order difference, a comprehensive model is established that captures the persistent influence of historical information. Through rigorous theoretical analysis, the fundamental properties including convergence and consensus is investigated. The results demonstrate that the proposed model not only maintains favorable convergence and consensus characteristics compared to classical opinion dynamics, but also addresses limitations such as the monotonicity of bounded opinions. This enables a more realistic representation of opinion evolution in real-world scenarios. The findings of this study provide new insights and methodological approaches for understanding opinion formation and evolution, offering both theoretical significance and practical applications.
\end{abstract}

\textbf{Keywords}: Memory, Fractional-Order, Hegselmann-Krause Model, Opinion Dynamics
\section{Introduction}

Opinion dynamics employs mathematical modeling, physical analogies, and computational methods to investigate the evolution of individual opinions during social interactions \cite{Castellano2009,Friedkin2015,Proskurnikov2017}. The field originates from the seminal work of French in 1956 \cite{french1956}, with subsequent developments bifurcating into two principal model classifications: those founded upon alternative interaction paradigms and those predicated on bounded confidence mechanisms.

Bounded confidence models, exemplified by the Deffuant-Weisbuch (DW) model \cite{deffuant2000mixing} and the Hegselmann-Krause (HK) model \cite{rainer2002opinion}, are grounded in the principle that agents selectively interact with neighbors whose opinions fall within a predefined similarity threshold \cite{Bernardo2024}. The HK model, in particular, captures polarization and consensus emergence through localized interactions and finite tolerance, making it well-suited for simulating opinion evolution in large-scale populations due to its self-organizing rules \cite{SU2017448}.
Other notable frameworks include the DeGroot model \cite{degroot1974reaching}, the Friedkin-Johnsen (FJ) model \cite{friedkin1990social} and the weighted-median model \cite{Mei2022}. These models excel in precisely describing pairwise influence mechanisms, rendering them ideal for analyzing opinion dynamics in specific groups with mathematical rigor.

Despite their theoretical sophistication, conventional opinion dynamics models face a fundamental limitation: they typically assume that opinion evolution depends solely on instantaneous interactions, neglecting the memory effects pervasive in real-world social behavior. In reality, individuals accumulate social experience through repeated interactions, and these historical impressions persistently shape future opinion formation. Empirical research robustly supports this view-psychological and sociological studies demonstrate that human decision-making is inherently history-dependent, with individuals often assessing current situations through the lens of past experiences. Such memory effects play a pivotal role in opinion dynamics, and interdisciplinary research further elucidates how memory shapes collective behavior, offering deeper insights into consensus formation, polarization and group decision-making \cite{gibson2005fine, mettke2003long, tsoar2011large}.

Recognizing this gap, recent studies have increasingly incorporated memory effects into opinion dynamics models. Some approaches, such as the FJ model, account for the persistence of initial opinions \cite{friedkin1990social}. Others introduce memory through mechanisms like time-delay effects \cite{Liu2023} or countdown-based interaction rules \cite{Jedrzejewski2018, Becchetti2023}.
However, existing models lack a unified mathematical framework for memory effects. Fractional-order calculus presents a promising solution, as it naturally captures non-local and history-dependent dynamics. Unlike integer-order models, fractional operators inherently weight past states, providing a principled way to integrate memory. Some studies have applied fractional calculus to opinion dynamics, such as \cite{girejko2014opinion, girejko2016, model2016hegselmann}, but these works adopt a direct replacement approach-converting standard models into fractional counterparts without addressing key limitations. For example, the cumulative weight of historical influences only asymptotically approaches unity over infinite time, violating realism in finite-step interactions (where total influence should sum to 1 at every step) and resulting in an insufficiency of theoretical analysis of dynamical properties of the fractional-order model.

This paper introduces a novel methodological framework that integrates fractional-order difference into the HK model to quantitatively capture memory effects in opinion formation. The HK model serves as a foundational paradigm in opinion dynamics, formalizing opinion exchange through its bounded confidence mechanism, where agents interact only with neighbors whose opinions fall within a predefined threshold. By leveraging local self-organizing rules, the HK model excels at simulating large-scale opinion evolution-a key advantage over global interaction models. Moreover, unlike prior studies that superficially impose fractional-order operators on classical models, our approach rigorously reformulates the HK framework using fractional calculus, ensuring that, the total historical influence sums to unity at each step (unlike ad hoc fractional extensions where weights only converge asymptotically).

The fractional-order bounded confidence model established in this paper exhibits strong theoretical properties. Through rigorous analysis, we prove some fundamental dynamical properties of the model, including
\begin{itemize}
        \item Maintains the ordinal consistency of opinions (order-preserving property);
        \item Guarantees convergence in isolated systems (asymptotic convergence);
        \item Achieves consensus when initial opinions are sufficiently concentrated, i.e., within mutual confidence bounds (consensus).
\end{itemize}

\noindent Crucially, our results reveal that memory effects fundamentally alter convergence properties:
\begin{itemize}
  \item Unlike classical HK models (finite-time convergence), the fractional variant only attains consensus asymptotically;
  \item The monotonicity of boundary opinions-a hallmark of the integer-order HK model-fails in the fractional case, reflecting memory's nonlinear smoothing effect.
\end{itemize}

\noindent These findings provide new theoretical insights into opinion dynamics, particularly how historical dependencies shape polarization, consensus, and collective decision-making in social networks.

The structure of this paper is as follows: Section 2 introduces fundamental concepts and discusses the main models examined in this work; Section 3 presents a detailed analysis of the proposed model's order-preserving property, convergence, and consensus conditions; Section 4 provides numerical simulation results; and Section 5 summarizes the findings and their implications.

\section{Model}\label{Mod_sec}
\renewcommand{\thesection}{\arabic{section}}

\subsection{Hegselmann-Krause model}
The HK model is a foundational framework in opinion dynamics that examines how interpersonal influence drives the evolution of opinions in social networks. In this model, each agent holds a continuous opinion value and updates it iteratively through interactions with neighboring agents. These interactions are constrained by bounded confidence: an agent can only influence or be influenced by others whose opinions fall within its confidence range. This mechanism, grounded in local rules and self-organization, gives the HK model distinct advantages for modeling opinion dynamics in large-scale populations. Based on this framework, the HK model can be formally defined as follows:

Let $\mathcal{V}=\{1, 2, \cdots, n\}$ be the set of agents, then
\begin{equation}\label{model:HKorig}
	x_i(k+1)=\frac{\sum_{j \in I_i(k)}x_j(k)}{|I_i(k)|},
\end{equation}
where $i \in \mathcal{V}$, $x_i(k) \in [0,1]$ is the opinion value of agent $i$ at time $k$ and
\begin{equation}
	I_{i}(k)=\Big\{j\in\mathcal{V}\Big||x_i(k)-x_j(k)|\leq\epsilon\Big\}
\end{equation}
is the neighbor set of agent $i$ at $k$ with $ \epsilon \in (0,1]$ representing the confidence threshold and $|\cdot|$ denoting the absolute value of a real number or the cardinality of a set accordingly.

\subsection{Fractional-order difference operator}

The integration of memory effects is crucial for modeling realistic agent behaviors in opinion dynamics. Building on dynamical systems theory, we utilize fractional-order derivatives to characterize agents' long-term memory properties. This framework captures temporal dependencies through which historical states continuously shape current opinion evolution.
Prior to model development, we introduce some preliminary materials.

Given $\alpha\in (0,1)$, define the following sequence
\begin{equation}\label{equa:akalpha}
	a_{k}^{(\alpha)}:=
	\begin{cases}
		1 &\text{for} \quad k=0 \\
	(-1)^k \frac{\alpha(\alpha-1) \cdots (\alpha-k+1)}{k!} &\text{for}\quad k\geq 1.
	\end{cases}.
\end{equation}
Since $( -1 ) ^k\frac{\alpha ( \alpha -1 ) \cdots	( \alpha -k+1 )}{k!}=\Big( \begin{array}{c}\alpha\\k\\	
 \end{array}\Big)$, the sequence $ ( a_{k}^{( \alpha )} )_{k\geq 0}$ can be rewritten using the generalized binomial as follows $a_{k}^{( \alpha )}=\left( -1 \right) ^k\Big( \begin{array}{c}\alpha\\k\\
 \end{array} \Big)$.
The sequence $( a_{k}^{( \alpha )})_{k\geq 1}$ has the following properties:
\begin{prop}\cite{podlubny1998}\label{prop:akalphadecreas}
For $ \alpha \in (0,1)$, $ a_{k}^{( \alpha )} <0 $ for $ k \geq 1$ and the sequence $( |a_{k}^{( \alpha )}|)_{k\geq 1}$ is decreasing with $k$. Moreover, $ \underset{k\rightarrow \infty}{\lim}a_{k}^{\left( \alpha \right)}=0$.
\end{prop}

\begin{prop} \cite{podlubny1998}
For $\alpha \in (0,1)$, the sequence $( |a_{k}^{( \alpha )}|)_{k\geq 1}$ satisfies:
\begin{equation}\label{equ:sumalphaleq1}
\begin{split}
  &\sum_{k=1}^{\infty}|a_{k}^{( \alpha )}|=1,\\ 
  &\sum_{k=1}^{K}|a_{k}^{( \alpha )}|<1, \,K\geq 0.
  \end{split}
\end{equation}
\end{prop}
Then one can define
\begin{defn} \cite{podlubny1998}
Let $\alpha \in \mathbb{R}$, $h>0$. The Gr\"unwald-Letnikov-type fractional-order difference operator $\Delta_{h}^\alpha$ of order $ \alpha$ for a function $y:\mathbb{N} \rightarrow \mathbb{R}$ is defined by
\begin{equation}\label{equ:glfrac}
(\Delta_{h}^\alpha y)(kh):=h^{-\alpha}\sum_{s=0}^{k} a_{s}^{(\alpha)} y(kh-sh),
\end{equation}
where $k \geq 0$ and $a_k^{(\alpha)}$ is the sequence given by (\ref{equa:akalpha}). In this paper, we will always consider $h=1$.
\end{defn}

\subsection{Memory-based bounded confidence model}

The conventional update mechanism in the classical HK model operates under the restrictive assumption that opinion evolution depends exclusively on the instantaneous neighborhood mean, thereby disregarding the crucial temporal dependencies inherent in real-world opinion formation processes. This formulation fails to capture the fundamental mnemonic characteristics properties that govern human decision-making dynamics in social systems.

To address this critical limitation, we propose a fractional-order calculus-based reformulation of the system in (\ref{model:HKorig}). For each agent $i\in\mathcal{V}$, the updated opinion state $x_i(k+1)$ at time $k+1$ depends not only on its current state $x_i(k)$ but also on its historical opinion trajectory (i.e., $x_i(s), s<k$). This temporal dependence is explicitly modeled via a Gr\"unwald-Letnikov-type fractional-order difference term ($h=1$ in (\ref{equ:glfrac})):
\begin{equation*}
  \sum_{s=0}^{k-1}|a_{k+1-s}^{(\alpha)}|x_{i}(s),
\end{equation*}
where $\alpha\in [0,1)$ governs the memory effect's decay rate.

To align with the mean-based update mechanism of the original HK model (\ref{model:HKorig}), we enforce a normalization condition ensuring the weights sum to unity. Consequently, the weight coefficient for the current opinion $x_i(k)$ is derived as:
\begin{equation}\label{equa:weightxik}
  1-\sum_{s=0}^{k-1}|a_{k+1-s}^{(\alpha)}|.
\end{equation}
Based on the above mechanism, denoting
\begin{equation}\label{model:neighfrac}
  \bar{I}_{i}(k)=I_{i}(k)/ \{i\},
\end{equation}
we propose the following fractional-order HK model:
\begin{equation}\label{model:HKfrac}
\begin{split}
	   x_{i}(k+1)=&\frac{\sum\limits_{j \in\bar{I}_{i}(k)}x_{j}(k)}{|I_{i}(k)|}+\frac{1}{|I_{i}(k)|}\Big[\sum_{s=0}^{k-1}|a_{k+1-s}^{(\alpha)}|x_{i}(s)+\Big(1-\sum_{s=0}^{k-1}|a_{k+1-s}^{(\alpha)}|\Big)x_{i}(k)\Big],
\end{split}	
\end{equation}
where
\begin{itemize}
    \item The first term aggregates the influence of neighboring opinions at time $k$.
    \item The second term models the cumulative memory effect of the agent's historical opinions (times 0 to $k-1$).
    \item The third term captures the immediate self-influence of the agent's current opinion (time $k$), weighted to ensure normalization.
\end{itemize}

The proposed model authentically captures the following hallmark features of real-world opinion dynamics:
\begin{enumerate}
  \item Initial Behavior ($k = 0$):

  \noindent When $k = 0$, the system (\ref{model:HKfrac}) is reduced to the classical HK model without memory. Here, $x_i(1)$ depends solely on neighbors and $x_i(0)$, with equal weight $\frac{1}{|I_i(0)|}$.
  \item Temporal Evolution ($k \geq 1$):
\begin{itemize}\label{comment2}
  \item The weight coefficient of the current opinion, given by $1 - \sum_{s=0}^{k-1} |a_{k+1-s}^{(\alpha)}|$, decays to $|a_1^{(\alpha)}|=\alpha$ as $k\to\infty$.
 This demonstrates that the fractional order $\alpha$ represents the minimal persistent influence of an agent's current opinion on its subsequent state, characterizing a fundamental lower bound in the system's memory-dependent dynamics.
  \item Conversely, the total historical influence grows as $\sum_{s=0}^{k-1} |a_{k+1-s}^{(\alpha)}|$ accumulates, and finally tends to $1-|a_1^{(\alpha)}|=1-\alpha$ as $k\to\infty$ (see (\ref{equ:sumalphaleq1})). Hence, larger $\alpha$ implies weaker memory effect.
\end{itemize}
  \item Long-Term Memory Decay:

\noindent By Proposition \ref{prop:akalphadecreas}, $|a_{k+1-s}^{(\alpha)}| \to 0$ for fixed $s$ as $k \to \infty$. Thus, distant memories (e.g., $x_i(s)$ for small $s$) gradually lose influence, reflecting fading recall in social learning.

\end{enumerate}

\section{Dynamical properties of the model}

In this section, we will demonstrate that the memory-driven model (\ref{model:HKfrac}) also preserves fundamental properties of opinion dynamics, including order preservation, asymptotic convergence, and consensus. Figure \ref{fig:basicevolu} illustrates a typical evolution pattern of the memory-enhanced model. Notably, unlike memoryless models (\ref{model:HKorig}) that achieve convergence in finite time, the memory-equipped model can only attain convergence asymptotically.

\begin{figure}[ht]
  \centering
  \includegraphics[width=3in]{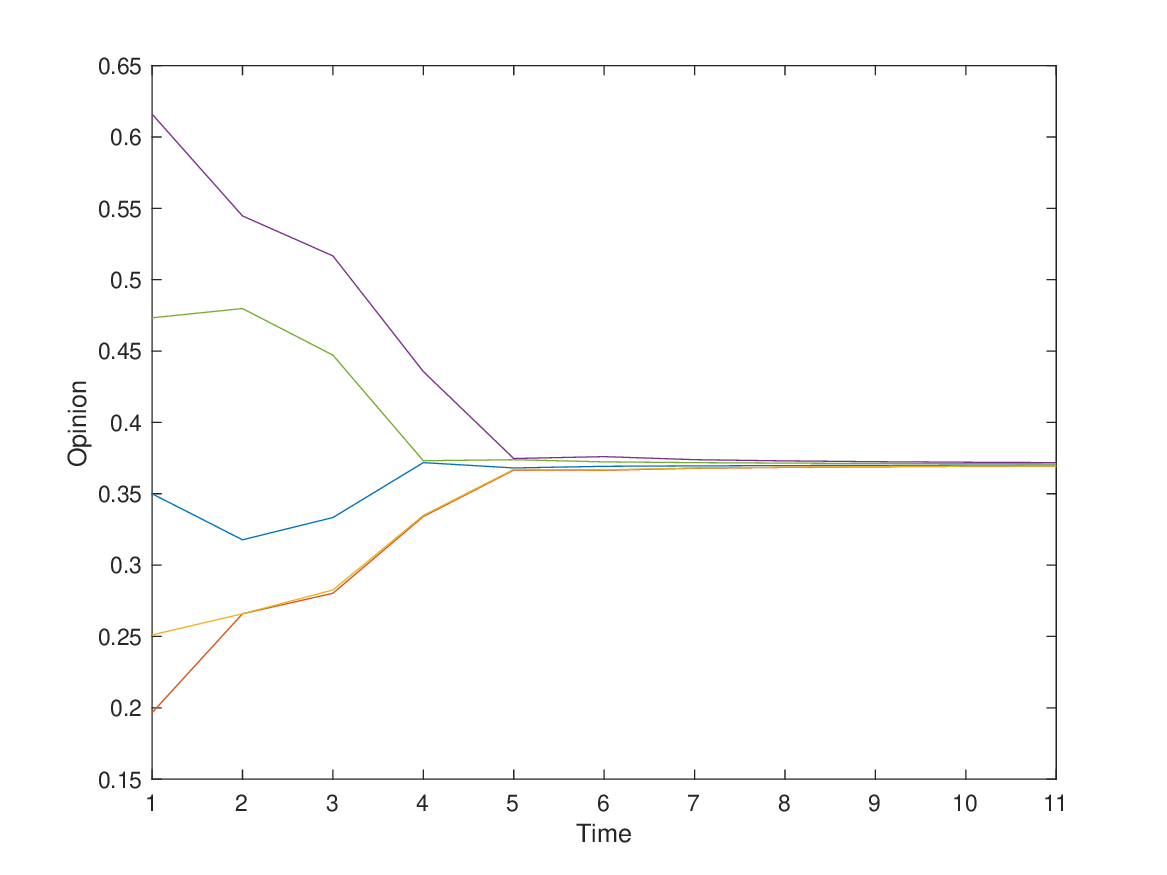}\\
  \caption{Five individuals with neighbor radius of $\epsilon=0.2$ and the  order is $\alpha=0.5$.}
  \label{fig:basicevolu}
\end{figure}

\subsection{Order-Preserving Property}
Our first theoretical result establishes that the system preserves opinion ordering throughout its evolution:
\begin{thm}(Order-Preserving Property) \label{thm:Order-Preserving Property}
For the system described by (\ref{model:HKfrac}), the opinion dynamics preserve initial ordering. Specifically, for any two agents $i, j\in \mathcal{V}$ with initial opinions satisfying $x_i(0)\leq x_j(0)$, their opinions maintain this ordering for all subsequent times:
\begin{equation}
x_i(k) \leq x_j(k), \quad  k \geq 0.
\end{equation}
\end{thm}

To establish the proof of Theorem \ref{thm:Order-Preserving Property}, the following basic lemma is needed.

\begin{lem} \cite{SU2017448}\label{lem:mean}
	Suppose $z_i, i=1,2,\ldots$ is a nonnegative nondecreasing (nonincreasing) sequence. Then for any $s \geq0$, the sequence $g_s(p)=\frac{1}{p}\sum_{i=s+1}^{s+p}z_i, p \geq 1$ is monotonically nondecreasing (nonincreasing) for $p$.
\end{lem}

\noindent\textbf{Proof of Theorem \ref{thm:Order-Preserving Property}:}
We will prove the conclusion using method of mathematical induction.

If $x_i(0)\leq x_j(0)$, by Lemma \ref{lem:mean} we have
\begin{equation}
	\begin{split}
x_{i}(1)=& \frac{\sum_{l\in \bar{I}_{i}(0)}x_{l}(0)}{|I_{i}(0)|}+\frac{1}{|I_{i}(0)|}x_{i}(0)\\
=& \frac{\sum_{l\in I_{i}(0)}x_{l}(0)}{|I_{i}(0)|}\\
\leq&  \frac{\sum_{l\in I_{j}(0)}x_{l}(0)}{|I_{j}(0)|}\\
=& \frac{\sum_{l\in \bar{I}_{j}(0)}x_{l}(0)}{|I_{j}(0)|}+\frac{1}{|I_{j}(0)|}x_{j}(0)\\
=&x_{j}(1).
\end{split}	
\end{equation}
Now suppose for $k \geq 1$, $x_i(s)\leq x_j(s)$ for all $s\leq k$. Then
\begin{equation}
	\begin{split}
&	\sum_{s=0}^{k-1}|a_{k+1-s}^{(\alpha)}|x_{i}(s)+\Big(1-\sum_{s=0}^{k-1}|a_{k+1-s}^{(\alpha)}|\Big)x_{i}(k)\leq \sum_{s=0}^{k-1}|a_{k+1-s}^{(\alpha)}|x_{j}(s)+\Big(1-\sum_{s=0}^{k-1}|a_{k+1-s}^{(\alpha)}|\Big)x_{j}(k).\\
	\end{split}
\end{equation}
Consequently, we obtain
\begin{equation}
	\begin{split}
		x_{i}(k+1)=& \frac{\sum_{l \in \bar{I}_{i}(k)}x_{l}(k)}{|I_{i}(k)|}+\frac{1}{|I_{i}(k)|}\Big[\sum_{s=0}^{k-1}|a_{k+1-s}^{(\alpha)}|x_{i}(s)+\Big(1-\sum_{s=0}^{k-1}|a_{k+1-s}^{(\alpha)}|\Big)x_{i}(k)\Big]\\
	\leq &  \frac{\sum_{l \in \bar{I}_{j}(k)}x_{l}(k)}{|I_{j}(k)|}+\frac{1}{|I_{j}(k)|}\Big[\sum_{s=0}^{k-1}|a_{k+1-s}^{(\alpha)}|x_{j}(s)+\Big(1-\sum_{s=0}^{k-1}|a_{k+1-s}^{(\alpha)}|\Big)x_{j}(k)\Big]\\
= &x_{j}(k+1).
	\end{split}
\end{equation}
Thus the method of mathematical induction yields the conclusion.\hfill $\Box$


\subsection{Asymptotic Convergence}
The order-preserving property established in Theorem \ref{thm:Order-Preserving Property} raises fundamental questions about the long-term evolution of the memory-enhanced system (\ref{model:HKfrac}). We now present the key convergence result for these dynamics:

\begin{thm}{(Asymptotic Convergence)} \label{thm:Convergence}
For the fractional opinion dynamics governed by (\ref{model:HKfrac}), the following holds for any initial configuration $x(0) \in [0,1]^n$:
\begin{itemize}
\item Each opinion trajectory admits a well-defined limit: For every agent $i \in \mathcal{V}$, there exists $x_i^* \in [0,1]$ such that $\lim_{k \to \infty} x_i(k) = x_i^*$;
\item The limiting configuration $x^* = (x_1^*,...,x_n^*)$ constitutes an equilibrium state of the system.
\end{itemize}
\end{thm}
\begin{proof}
First we consider the maximum opinion value $x_M(k)=\max_{i\in\mathcal{V}}x_i(k)$, $k\geq 0$. By the order-preserving property revealed in Theorem \ref{thm:Order-Preserving Property}, we can denote $x_1(k)=x_M(k)$ without loss of generality.
We will show that $x_1(k) \leq x_1(0)$ for all $k>0$.

For $k=0$, the definition of $x_1(k)$ implies, for all $ i \in \mathcal{V}$ and $k>0$, $x_i(0)\leq x_1(0)$,
thus,
\begin{equation} \label{equ:xioayu1}
	\frac{\sum_{j \in I_1(0)}x_j(0)}{|I_1(0)|} \leq x_1(0).
\end{equation}
When $k=0$, we have $\displaystyle\sum_{s=0}^{k-1}|a_{k+1-s}^{(\alpha)}|=0$, so the model (\ref{model:HKfrac}) degenerates to the model (\ref{model:HKorig}) (the classical HK model), which implies
\begin{equation} \label{equ:0-1}
	x_1(1)=\frac{\sum_{j \in I_1(0)}x_j(0)}{|I_1(0)|}\leq x_1(0).
\end{equation}
For $k=1$, the definition of $x_1(k)$ and (\ref{equ:0-1}) imply
\begin{equation} \label{equ:xiaoyu2}
	\frac{\sum\limits_{j \in \bar{I}_1(1)}x_j(1)}{|I_1(1)|}\leq \frac{\sum\limits_{j \in \bar{I}_1(1)}x_1(1)}{|I_1(1)|}\leq \frac{\sum\limits_{j \in \bar{I}_1(1)}x_1(0)}{|I_1(1)|}.
\end{equation}
Since $|a_m^{(\alpha)}|<1$ for each $m\geq 0$, by (\ref{equ:0-1}) and (\ref{equ:xiaoyu2}) we have
\begin{equation*}
\begin{split}
	x_1(2)=&\frac{\sum_{j \in \bar{I}_1(1)}x_1(1)}{|I_1(1)|}+\frac{1}{|I_1(1)|}\Big[|a_2^{(\alpha)}|x_1(0)+\Big(1-|a_2^{(\alpha)}|\Big)x_1(1)\Big]\\
	\leq &\frac{\sum_{j \in \bar{I}_1(1)}x_1(0)}{|I_1(1)|}+\frac{1}{|I_1(1)|}\Big[|a_2^{(\alpha)}|x_1(0)+\Big(1-|a_2^{(\alpha)}|\Big)x_1(0)\Big]\\
	=&x_1(0).
\end{split}
\end{equation*}
Applying recursive reasoning we can show that
\begin{equation}
	x_1(k)\leq x_1(0), \quad k> 0.
\end{equation}
Next we will show that there exists a sequence $0<K_{1}<K_{2}<K_{3}< \dots$ such that
\begin{equation}\label{equa:ordermax}
 	x_1(0) \geq x_1(K_{1}) \geq x_1(K_{2}) \geq \ldots.
 \end{equation}
We first demonstrate the existence of $K_1>0$ such that
\begin{equation}\label{equ:existK1}
	x_1(k) \leq x_1(K_1), \quad k>K_1.
\end{equation}
By Proposition \ref{prop:akalphadecreas} ($\lim\limits_{k \rightarrow \infty}|a_k^{(\alpha)}|=0$), we know that given any $T>0$ and $\delta>0$, we can find $k^{(T)}>T$ (one can also observe that $k^{(T)}-T$ only depends on $\delta$) such that
\begin{equation}\label{equ:alphaleqdelta}
	|a_{k^{(T)}+1-s}^{(\alpha)}| < \frac{\delta}{T+1}, \quad \forall s \leq T.
\end{equation}
Since $0\leq x_i(k)\leq1$, we have
\begin{equation} \label{equ:wuxiangxiao}
	\sum_{s=0}^{T}|a_{k^{(T)}+1-s}^{(\alpha)}|x_{i}(s) \leq \delta.
\end{equation}
For $k>k^{(T)}$, the system (\ref{model:HKfrac}) can be rewritten as
\begin{equation}\label{equ:gaixie}
	\begin{split}
		x_{i}(k+1)= & \frac{\sum_{j \in \bar{I}_{i}(k)}x_{j}(k)}{|I_{i}(k)|}\\
&+\frac{1}{|I_{i}(k)|}\Big[\sum_{s=0}^T|a_{k+1-s}^{(\alpha)}|x_{i}(s)+\sum_{s=T+1}^{k-1}|a_{k+1-s}^{(\alpha)}|x_{i}(s)+\Big(1-\sum_{s=0}^{k-1}|a_{k+1-s}^{(\alpha)}|\Big)x_{i}(k)\Big].
\end{split}
\end{equation}
Denote $t_M:=\{T+1\leq k\leq k^{(T)}|x_1(k)\geq x_1(l), T+1\leq l\leq k^{(T)}\}$, $t_m:=\{T+1\leq k\leq k^{(T)}|x_1(k)\leq x_1(l), T+1\leq l\leq k^{(T)}\}$. Given $T>0$, if for any $\delta>0$, the above selected $x_1(t_M)-x_1(t_m)<\frac{1}{|a_{k^{(T)}-T}^{(\alpha)}|} \delta$, we know
\begin{equation}\label{equ:x1MmkT}
  x_1(t_M)=x_1(t_m)=x_1(k),\quad T+1\leq k\leq k^{(T)}.
\end{equation}
If (\ref{equ:x1MmkT}) holds for all $T>0$, (\ref{equ:existK1}) and also (\ref{equa:ordermax}) hold. Otherwise, there exists $T>0$ and  $\delta>0$ such that
\begin{equation}\label{equ:x1Mmgeq}
  x_1(t_M)-x_1(t_m)> \frac{1}{|a_{k^{(T)}-T}^{(\alpha)}|} \delta.
\end{equation}
Since $\{|a_m^{(\alpha)}|, m\geq 0\}$ is decreasing by Proposition \ref{prop:akalphadecreas},
\begin{equation}\label{equ:alphaleq1}
 \frac{|a_{k^{(T)}+s-t_m}^{(\alpha)}|}{|a_{k^{(T)}-T}^{(\alpha)}|}\geq 1, \quad s\geq 1.
\end{equation}
Let $\bar{a}_M=1-\sum_{s=0}^{M-1}|a_{M+1-s}^{(\alpha)}|$, by (\ref{equ:sumalphaleq1}), (\ref{equ:wuxiangxiao}) and (\ref{equ:gaixie}),
\begin{equation} \label{equ:zhuyao0}
	\begin{split}
x_1(k^{(T)}+1)=&\frac{\sum\limits_{j \in \bar{I}_1(k^{(T)})}x_j(k^{(T)})}{|I_1(k^{(T)})|}+\frac{1}{|I_1(k^{(T)})|}\Big[\sum_{s=0}^T|a_{k^{(T)}+1-s}^{(\alpha)}| x_1(s)\\
&+\sum_{s=T+1}^{k^{(T)}-1}|a_{k^{(T)}+1-s}^{(\alpha)}|x_1(s)+\Big(1-\sum_{s=0}^{k^{(T)}-1}|a_{k^{(T)}+1-s}^{(\alpha)}|\Big)x_1(k^{(T)})\Big]\\
\leq &\frac{\sum\limits_{j \in \bar{I}_1(k^{(T)})}x_j(k^{(T)})}{|I_1(k^{(T)})|}+\frac{1}{|I_1(k^{(T)})|}\Big(\delta+\sum_{s=T+1}^{k^{(T)}-1}|a_{k^{(T)}+1-s}^{(\alpha)}|x_1(s)+\bar{a}_{k^{(T)}}x_1(k^{(T)})\Big).
\end{split}
\end{equation}
Further by (\ref{equ:x1Mmgeq}) and (\ref{equ:alphaleq1}), we obtain
\begin{equation} \label{equ:zhuyao}
\begin{split}
x_1(k^{(T)}+1)<&\frac{\sum\limits_{j \in \bar{I}_1(k^{(T)})}x_1(k^{(T)})}{|I_1(k^{(T)})|}+\frac{1}{|I_1(k^{(T)})|}\Big(\delta-\frac{|a_{k^{(T)}+1-t_m}^{(\alpha)}|}{|a_{k^{(T)}-T}^{(\alpha)}|}\delta\\
&+\sum_{s=T+1}^{k^{(T)}-1}|a_{k^{(T)}+1-s}^{(\alpha)}|x_1(t_M)+\bar{a}_{k^{(T)}}x_1(t_M)\Big)\\
< &\frac{\sum\limits_{j \in \bar{I}_1(k^{(T)})}x_1(t_M)}{|I_1(k^{(T)})|}+\frac{1}{|I_1(k^{(T)})|}x_1(t_M)\\
=&x_1(t_M).
\end{split}	
\end{equation}
By (\ref{equ:alphaleq1}) and (\ref{equ:zhuyao}),
\begin{equation}\label{equ:zhuyao21}
	\begin{split}
		x_1(k^{(T)}+2)=&\frac{\sum\limits_{j \in \bar{I}_1(k^{(T)}+1)}x_j(k^{(T)}+1)}{|I_1(k^{(T)}+1)|}+\frac{1}{|I_1(k^{(T)}+1)|}\Big[\sum_{s=0}^T|a_{k^{(T)}+2-s}^{(\alpha)}|x_1(s)+\sum_{s=T+1}^{k^{(T)}}|a_{k^{(T)}+2-s}^{(\alpha)}|x_1(s)\\
&+\Big(1-\sum_{s=0}^{k^{(T)}}|a_{k^{(T)}+2-s}^{(\alpha)}|\Big)x_1(k^{(T)}+1)\Big]\\
< &\frac{\sum\limits_{j \in \bar{I}_1(k^{(T)}+1)}x_1(t_M)}{|I_1(k^{(T)}+1)|}+\frac{1}{|I_1(k^{(T)}+1)|}\Big(\delta+\sum_{s=T+1}^{k^{(T)}}|a_{k^{(T)}+2-s}^{(\alpha)}|x_1(s)+\bar{a}_{k^{(T)}+1}x_1(k^{(T)}+1)\Big).
\end{split}
\end{equation}
By (\ref{equ:x1Mmgeq}), (\ref{equ:alphaleq1}) and (\ref{equ:zhuyao21}),
\begin{equation}\label{equ:zhuyao2}
	\begin{split}
		x_1(k^{(T)}+2)<&\frac{\sum\limits_{j \in \bar{I}_1(k^{(T)}+1)}x_1(t_M)}{|I_1(k^{(T)}+1)|}+\frac{1}{|I_1(k^{(T)}+1)|}\Big(\delta-\frac{|a_{k^{(T)}+2-t_m}^{(\alpha)}|}{|a_{k^{(T)}-T}^{(\alpha)}|}\delta\\
&+\sum_{s=T+1}^{k^{(T)}}|a_{k^{(T)}+2-s}^{(\alpha)}|x_1(t_M)+\bar{a}_{k^{(T)}+1}x_1(t_M)\Big)\\
< &\frac{\sum\limits_{j \in \bar{I}_1(k^{(T)}+1)}x_1(t_M)}{|I_1(k^{(T)}+1)|}+\frac{1}{|I_1(k^{(T)}+1)|}x_1(t_M)\\
=&x_1(t_M).
\end{split}	
\end{equation}
Applying recursive method we obtain
\begin{equation*}
  x_1(t_M)\geq x_1(k),\quad k\geq t_M.
\end{equation*}
Take $K_1=t_M$, and we know that (\ref{equ:existK1}) holds. Repeating the above procedure by considering the case $T>K_1$, we can prove that there exists $K_2>K_1$ that
\begin{equation*}
  x_1(k)\leq x_1(K_2)\leq x_1(K_1),\quad k\geq K_2.
\end{equation*}
Iterating this argument, we show that there exists a sequence $K_{1}<K_{2}<K_{3}< \dots$ such that
\begin{equation}\label{equ:Ksleq}
  x_1(k)\leq x_1(K_l), \quad k\geq K_l,\quad l=1,2,\dots
\end{equation}
and
\begin{equation*}
 	x_1(0) \geq x_1(K_{1}) \geq x_1(K_{2}) \geq \dots.
 \end{equation*}
The sequence $\{x_1(K_l), l=1,2,\dots\}$ is monotonically decreasing and bounded below, so there exists $x_1^*\in [0,1]$ such that
 \begin{equation}\label{equ:limitx1}
   \lim_{l\to \infty}x_1(K_l)=x_1^*.
 \end{equation}
This also implies
\begin{equation}\label{equ:limitx1k}
  \lim\limits_{k \rightarrow \infty}x_1(k)=\lim\limits_{l \rightarrow \infty}x_1(K_l)=x_1^*.
\end{equation}
If not, we have $\varliminf \limits_{k \rightarrow \infty}x_1(k)< \lim \limits_{l \rightarrow \infty}x_{M}(K_l)=x_1^*$.
Hence there exists $\rho>0$ such that for any $t>0$ there exists $k_t^\rho>t$ such that
\begin{equation}\label{equ:x1ktrho}
  x_1(k_t^\rho)<x_1^*-\rho.
\end{equation}
Meanwhile, by (\ref{equ:Ksleq}) and (\ref{equ:limitx1}), we know for any $\epsilon>0$, there exists $T_\epsilon>0$ such that
\begin{equation}\label{equ:Kepsi}
  x_1(k)<x_1^*+\epsilon,\quad k>T_\epsilon.
\end{equation}
By Proposition \ref{prop:akalphadecreas} and (\ref{equ:alphaleqdelta}), given any $\delta>0$, $T>0$, there exists $k^{(T)}>0$ such that for $k>k^{(T)}$, (\ref{equ:alphaleqdelta}) holds and moreover, $k^{(T)}-T$ only depends on $\delta$. So we can denote $L(\delta)=k^{(T)}-T$.
Take $\delta=\frac{\alpha\rho}{2}$ in (\ref{equ:alphaleqdelta}), $\epsilon=\frac{\alpha\rho}{2L(\delta)}$ in (\ref{equ:Kepsi}), $T$ in (\ref{equ:alphaleqdelta}) is then taken to be $T_\epsilon$ in (\ref{equ:Kepsi}). Also by (\ref{equ:x1ktrho}), there exists $k_{T_\epsilon}^\rho\geq k^{(T_\epsilon)}$ such that
\begin{equation}\label{equ:x1ktepsrho}
  x_1(k_{T_\epsilon}^\rho)<x_1^*-\rho.
\end{equation}
By (\ref{equ:alphaleqdelta}), (\ref{equ:Kepsi}) and (\ref{equ:x1ktepsrho}),
\begin{equation}\label{equ:x1ktepsrhoplus1}
	\begin{split}
		x_1(k_{T_\epsilon}^\rho+1)=&\frac{\sum\limits_{j \in \bar{I}_1(k_{T_\epsilon}^\rho)}x_j(k_{T_\epsilon}^\rho)}{|I_1(k_{T_\epsilon}^\rho)|}+\frac{1}{|I_1(k_{T_\epsilon}^\rho)|}\Big(\sum_{s=0}^{T_\epsilon}|a_{k_{T_\epsilon}^\rho+1-s}^{(\alpha)}|x_1(s)\\ &+\sum_{s=T_\epsilon+1}^{k_{T_\epsilon}^\rho-1}|a_{k_{T_\epsilon}^\rho+1-s}^{(\alpha)}|x_1(s)+\bar{a}_{k_{T_\epsilon}^\rho}x_1(k_{T_\epsilon}^\rho)\Big)\\
< &\frac{\sum\limits_{j \in \bar{I}_1(k_{T_\epsilon}^\rho)}x_1(k_{T_\epsilon}^\rho)}{|I_1(k_{T_\epsilon}^\rho)|}+\frac{1}{|I_1(k_{T_\epsilon}^\rho)|}\Big(\delta+\sum_{s=T_\epsilon+1}^{k_{T_\epsilon}^\rho-1}|a_{k_{T_\epsilon}^\rho+1-s}^{(\alpha)}|(x_1^*+\epsilon)+\bar{a}_{k_{T_\epsilon}^\rho}(x_1^*-\rho)\Big)\\
\leq&(x_1^*-\rho)+\frac{1}{|I_1(k_{T_\epsilon}^\rho)|}\Big(\delta+\sum_{s=T_\epsilon+1}^{k_{T_\epsilon}^\rho-1}|a_{k_{T_\epsilon}^\rho+1-s}^{(\alpha)}|(\epsilon+\rho)\Big)\\
\leq &x_1^*-\Big(1-\sum_{s=T_\epsilon+1}^{k_{T_\epsilon}^\rho-1}|a_{k_{T_\epsilon}^\rho+1-s}^{(\alpha)}|\Big)\rho+\delta+L(\delta)\epsilon.
\end{split}		
\end{equation}
Since $1-\sum_{s=T_\epsilon+1}^{k_{T_\epsilon}^\rho-1}|a_{k_{T_\epsilon}^\rho+1-s}^{(\alpha)}|\geq \alpha$ by the comment (\ref{comment2}) behind the  system (\ref{model:HKfrac}), (\ref{equ:x1ktepsrhoplus1}) yields
\begin{equation}\label{equ:x1ktepsrhoplus12}
	\begin{split}
		x_1(k_{T_\epsilon}^\rho+1)< &x_1^*-\alpha\rho+\delta+L(\delta)\epsilon<x_1^*.
\end{split}		
\end{equation}
Following a similar argument of (\ref{equ:zhuyao2}), we can show that
\begin{equation}\label{equ:x1ktepsrhoall}
  x_1(k)<x_1^*,\quad k>k_{T_\epsilon}^\rho,
\end{equation}
which contradicts (\ref{equ:limitx1}), and hence (\ref{equ:limitx1k}) holds.

Let $\mathcal{V}_1:=\{i \in \mathcal{V}|\lim\limits_{k \rightarrow \infty}x_i(k)=x_1^*\}$. If $j \notin \mathcal{V}_1$, it is easy to show that $j$ cannot be the neighbor of agents in $\mathcal{V}_1$ (or their opinion values have to be less than $x_1^*$). Then for $\mathcal{V}/\mathcal{V}_1$, we can similarly define $x_1(k)$ and a same method leads to
\begin{equation}
	\lim\limits_{k \rightarrow \infty}x_1(k)=x^*_2.
\end{equation}
where $x^*_2\in[0,1]$ and $|x_1^*-x^*_2|>\epsilon$.

By repeating the above procedure, we can conclude that the opinion values of all individuals asymptotically converge. Clearly, $x^* = (x_1^*,...,x_n^*)$ constitutes an equilibrium state of the system.
\end{proof}

\begin{rem}
The convergence analysis becomes significantly more challenging in the fractional-order case compared to its integer-order counterpart. Unlike classical HK dynamics where boundary opinions evolve monotonically, the memory effects in fractional-order systems destroy this crucial monotonicity property. This technical obstacle is further compounded by the inherent complexity of bounded-confidence models, where extreme nonlinear interactions already make convergence proofs notoriously difficult - even the convergence problem for integer-order heterogeneous models remains open to date. To overcome these challenges, our proof strategy involves constructing carefully designed asymptotic subsequences and developing inductive techniques to properly account for historical state dependencies while preventing potential oscillatory behaviors.
\end{rem}
The following simple example demonstrates the loss of monotonicity in boundary opinion evolution:

\begin{example}\label{exam:nonmono}
For the system (\ref{model:HKfrac}), given the initial opinion values $x(0)=(1.0944, 0.2772)$ and the neighbor radius $\epsilon=1$, we obtain the following:
\begin{gather*}
	x(1)=(0.6823, 0.6834)\\
	x(2)=(0.7079, 0.6591)\\
	x(3)=(0.6941, 0.6735)\\
	x(4)=(0.6938, 0.6766)\\
	x(5)=(0.6858, 0.6788)\\
	x(6)=(0.6869, 0.6791).
\end{gather*}
\end{example}
Compared to the classical HK model where maximum opinion values evolve in a monotonically non-increasing manner, we can see the model (\ref{model:HKfrac}) in Example \ref{exam:nonmono} exhibits non-monotonic fluctuations with occasional upward perturbations. This phenomena can also be observed in Figures \ref{fig:basicevolu}-\ref{fig:consensus}. This more realistically captures phenomena such as opinion rebounds due to memory effect, local reinforcement effects, and cognitive updates--reflecting the inherent nonlinearity of real-world opinion dynamics.

\subsection{Consensus Criterion}
In opinion dynamics research, the mathematical characterization of consensus conditions remains a fundamental theoretical problem. Based on Theorem \ref{thm:Convergence}, we establish the following consensus theorem:
\begin{thm} \label{thm:consensus}
Let $x_M(k)=\max_{i\in\mathcal{V}}x_i(k)$, $x_m(k)=\min_{i\in\mathcal{V}}x_i(k)$ and $d(k)=x_M(k)-x_m(k)$ for $k\geq 0$. Given $\epsilon\in(0,1]$, suppose the initial opinion value satisfies $d(0) \leq \epsilon$. Then the system (\ref{model:HKfrac}) asymptotically achieves consensus, i.e.,
	\begin{equation*}
		\lim_{k \rightarrow \infty}d(k)=0.
	\end{equation*}
\end{thm}
\begin{rem}
This theorem demonstrates that the system (\ref{model:HKfrac}) asymptotically reaches consensus under initially fully connected conditions. This conclusion verifies that the fractional-order dynamical model proposed in this work preserves the fundamental properties of classical consensus theory, while providing a theoretical foundation for subsequent studies on consensus conditions under more situations.
\end{rem}
To establish Theorem \ref{thm:consensus}, we first prove that if all individuals are initially mutual neighbors, this connectivity persists indefinitely.
\begin{lem}\label{lem:neigh}
Given $\epsilon\in(0,1]$, suppose the initial opinion value satisfies $|x_M(0)-x_m(0)| \leq \epsilon$, then $|x_M(k)-x_m(k)|\leq \epsilon$ for all $k\geq 0$.
\end{lem}
\begin{proof}
By (\ref{equa:ordermax}), we establish that for all discrete time instances $k \geq 0$, the opinion spectrum of the dynamical system (\ref{model:HKfrac}) obeys the inequality:
	\begin{equation*}
		|x_M(k)-x_m(k)| \leq |x_M(0)-x_m(0)| \leq \epsilon.
	\end{equation*}
\end{proof}

\noindent\textbf{Proof of Theorem \ref{thm:consensus}:}
The complete connectivity case represents a special instance of Theorem \ref{thm:Convergence}. When all individuals mutually interact, Lemma \ref{lem:neigh} guarantees maintained connectivity, and Theorem \ref{thm:Convergence} ensures convergence to a shared limit value, resulting in consensus.\hfill $\Box$

\section{Simulation}

In this section, we present simulation results that validate the theoretical analysis of the fractional-order opinion dynamics model.

The previous Figure \ref{fig:basicevolu} presents the opinion dynamics of a $5$-agent system with interaction radius $\epsilon=0.2$ and fractional order $\alpha=0.5$.
The simulation results confirm the order-preserving property of the system: for all time steps $k \geq 0$, the inequality $x_i(k) \geq x_j(k)$ holds whenever $x_i(0) \geq x_j(0) $.

Figure \ref{fig:convergence} displays the fragmentation pattern of a similar $5$-agent system ($\epsilon=0.2, \alpha=0.5$). The simulation shows that the system exhibits a fragmented convergence pattern, where distinct clusters emerge with intra-cluster consensus. Notably, each individual's opinion value asymptotically converges to a stable equilibrium point, demonstrating that while global consensus is not achieved, local convergence behavior is clearly observable within subsystem clusters.

\begin{figure}[ht]
  \centering
  \includegraphics[width=3in]{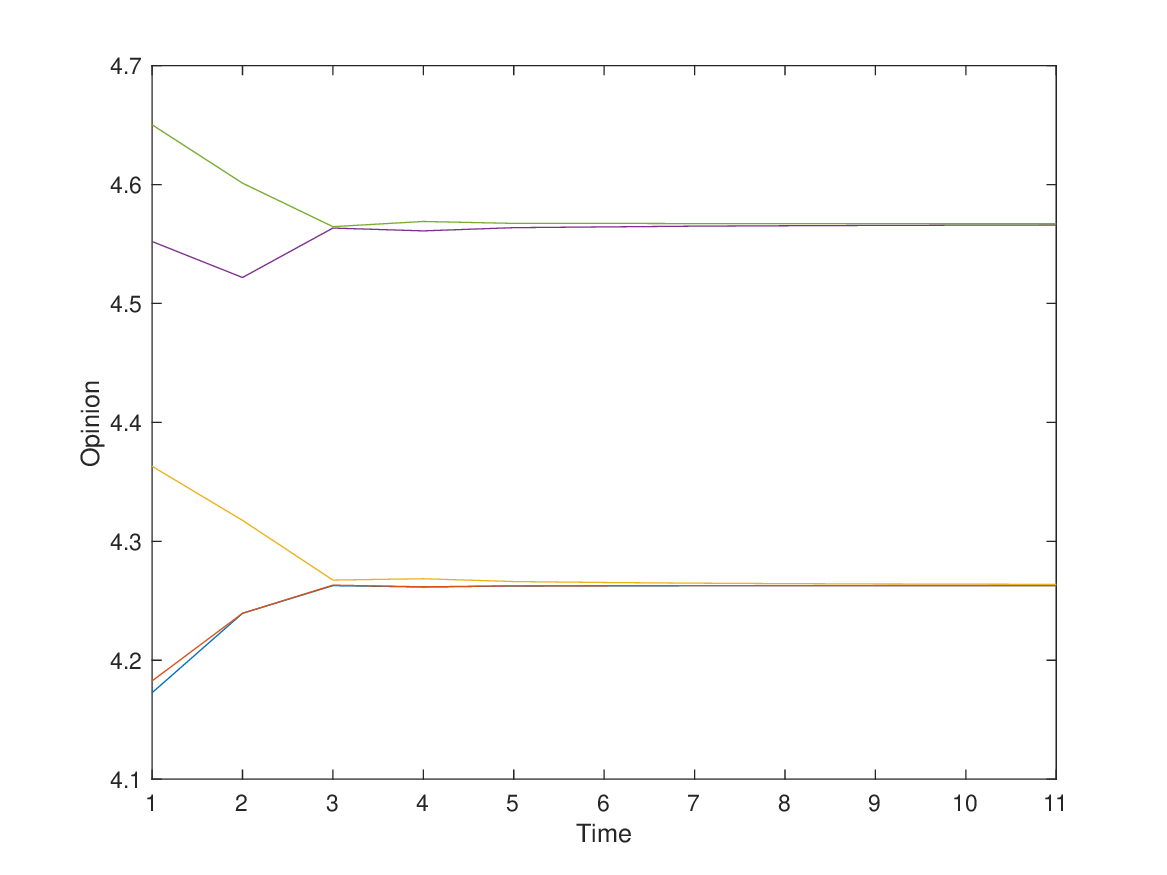}\\
  \caption{Five individuals with neighbor radius of $\epsilon=0.2$ and the  order is $\alpha=0.5$.}
  \label{fig:convergence}
\end{figure}

Figure \ref{fig:consensus} shows consensus achievement in a $4$-agent system with expanded interaction range ($\epsilon=1, \alpha=0.5$). The simulation results verify that when all initial opinions fall within the neighborhood threshold, the fractional-order HK model (\ref{model:HKfrac}) guarantees asymptotic global consensus.

\begin{figure}[ht]
  \centering
  \includegraphics[width=3in]{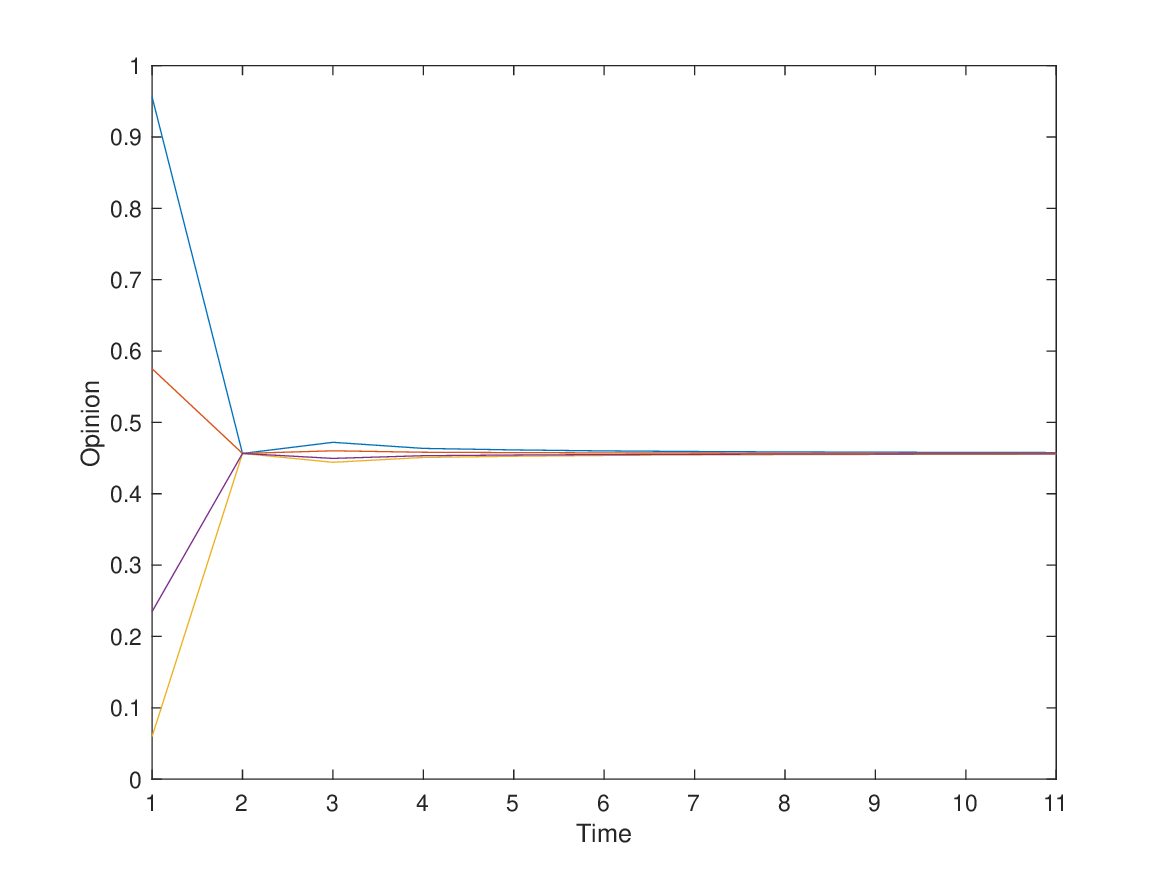}\\
  \caption{Four individuals with neighbor radius of $\epsilon=1$ and the  order is $\alpha=0.5$.}\label{fig:consensus}
\end{figure}
\section{Conclusion}

This paper explores the ``memory'' effect in opinion dynamics and proposes a novel fractional-order model based on the classical Hegselmann-Krause framework. By incorporating fractional-order difference, the model more accurately reflects the influence of historical experiences on individual opinion adjustments, aligning closely with real-world social interactions. Theoretical analysis demonstrates that the model ensures opinion convergence without external interference while preserving the order of individual opinions. This provides new insights into group consensus mechanisms, particularly when initial opinions are similar. The proposed framework naturally incorporates subtle opinion variations in early stages of evolution, faithfully representing the nonlinear complexities of actual opinion formation processes, unlike conventional integer-order approaches.

This research provides a solid foundation for advancing opinion dynamics studies. Future work could explore the interaction between memory effects and external factors like media or social interventions to improve predictive accuracy. The model's ability to capture nonlinear and lagged effects in collective behavior holds potential for applications in public opinion forecasting and group decision-making. Extending this framework to multi-dimensional opinion spaces or heterogeneous networks also offers promising avenues for further exploration.



\begin{thebibliography}{99}

\bibitem{Castellano2009}
C. Castellano, S. Fortunato, and V. Loreto, Statistical physics of social dynamics, {\it Rev. Mod. Phys.}, vol.81, pp. 591--646, 2009.

\bibitem{Friedkin2015}
N. Friedkin, The problem of social control and coordination of complex systems in sociology: a look at the community cleavage, {\it IEEE Control Systems}, 35(3), pp. 40--51, 2015.

\bibitem{Proskurnikov2017}
A. Proskurnikov, R. Tempo, A Tutorial on Modeling and Analysis of Dynamic Social Networks. Part I., {\it Annual Reviews in Control}, vol. 43, pp. 65--79, 2017.

\bibitem{french1956}
 J. RP French, A formal theory of social power, Psychological review, 63.3 (1956), pp. 181-194.

\bibitem{deffuant2000mixing}
Guillaume Deffuant et al. ``Mixing beliefs among interacting agents'', \emph{Advances in Complex Systems} 3.01n04 (2000), pp. 87-98.
%

\bibitem{rainer2002opinion}
Hegselmann, Rainer and Krause, Ulrich. ``Opinion dynamics and bounded confidence: models, analysis and simulation'', \emph{Journal of Artificial Societies and Social Simulation}, 5.3 (2002), pp. 1-33.

\bibitem{Bernardo2024}
C. Bernardo, C. Altafini, A. Proskurnikov, F. Vasca, Bounded confidence opinion dynamics: A survey, Automatica, 159, 111302, 2024.

\bibitem{SU2017448}
Wei Su, Ge Chen, and Yiguang Hong. ``Noise leads to quasi-consensus of Hegselmann-Krause opinion dynamics'', \emph{Automatica} 85 (2017), pp. 448-454.

\bibitem{degroot1974reaching}
Morris H DeGroot. ``Reaching a consensus'', \emph{Journal of the American Statistical association} 69.345 (1974), pp. 118-121.

\bibitem{friedkin1990social}
Noah E Friedkin and Eugene C Johnsen. ``Social influence and opinions'', \emph{Journal of mathematical sociology} 15.3-4 (1990), pp. 193-205.

\bibitem{Mei2022}
W. Mei, F. Bullo, G. Chen, J. M. Hendrickx, and F. Dorfler, Micro-foundation of opinion dynamics: Rich consequences of the weighted-median mechanism, Physical Review Research, vol. 4, no. 2, p. 023213, 2022

%
%

\bibitem{gibson2005fine}
Brett M Gibson and Alan C Kamil. ``The fine-grained spatial abilities of three seedcaching corvids'',\emph{Learning \textup{\&} behavior} 33.1 (2005), pp. 59-66.

\bibitem{mettke2003long}
Claudia Mettke-Hofmann and Eberhard Gwinner, ``Long-term memory for a life on the move'', \emph{Proceedings of the National Academy of Sciences} 100.10 (2003), pp. 5863-5866.

\bibitem{tsoar2011large}
Asaf Tsoar et al. ``Large-scale navigational map in a mammal''. \emph{Proceedings of the National Academy of Sciences} 108.37 (2011), pp: 718-724.

\bibitem{Liu2023}
Q. Liu and L. Chai, The memory influence on opinion dynamics in coopetitive social networks: Analysis, application, and simulation, IEEE Trans. Control Netw. Syst., vol. 10, no. 4, pp. 1867-1878, 2023.

\bibitem{Jedrzejewski2018}
A. Jedrzejewski and K. Sznajd-Weron, Impact of memory on opinion dynamics, Physica A, vol. 505, pp. 306-315, 2018.

\bibitem{Becchetti2023}
L. Becchetti, A. Clementi, A. Korman, F. Pasquale, L. Trevisan, and R. Vacus, On the role of memory in robust opinion dynamics, in Proc 32nd Int. Joint Conf. Artif. Intell., 2023.

\bibitem{girejko2014opinion}
Ewa Girejko and Dorota Mozyrska. ``Opinion dynamics and fractional operators'', \emph{ICFDA'14 International Conference on Differentiation and Its Applications 2014.} IEEE. (2014), pp. 1-4.
%
%
%


\bibitem{girejko2016}
	Ewa Girejko, Dorota Mozyrska, and Ma{\l}gorzata Wyrwas. ``On the  Continuous-Time Hegselmann-Krause's Type Consensus Model'', \emph{Theory and Applications of Non-integer Order Systems: 8th Conference on Non-integer Order Calculus and Its Applications, Zakopane, Poland}. Springer. (2016), pp. 21-32.

\bibitem{model2016hegselmann}
	Dorota Mozyrska, and Malgorzata Wyrwas. ``Fractional discrete-time of Hegselmann-Krause's type consensus model with numerical simulations''. \emph{Neurocomputing}, 216 (2016), pp. 381--392.

\bibitem{podlubny1998}
Igor Podlubny. ``differential equations: an introduction to derivatives,  differential equations, to methods of their solution and some of their applications''. Vol. 198. elsevier, 1998.

\end{thebibliography}
\end{document}